\documentclass[12pt]{article}
\usepackage{amssymb,amsmath,xcolor,graphicx,xspace,colortbl,rotating}
\usepackage{amsmath}
\usepackage{amssymb}
\usepackage{amssymb,amsmath,xcolor,graphicx,xspace,colortbl,rotating}
\usepackage{graphicx}

\graphicspath{{Trust_and_group formation BC 030820sbv2_graphics/}{Trust_and_group formation BC 030820sbv2_tcache/}{Trust_and_group formation BC 030820sbv2_gcache/}}
\DeclareGraphicsExtensions{.pdf,.eps,.ps,.png,.jpg,.jpeg}
\graphicspath{{Trust_and group formation BC 030820sb_graphics/}{Trust_and group formation BC 030820sb_tcache/}{Trust_and group formation BC 030820sb_gcache/}}
\DeclareGraphicsExtensions{.pdf,.eps,.ps,.png,.jpg,.jpeg}

\newtheorem{proposition}{{\bf \sc Proposition}}
\newtheorem{remark}{{\bf \sc Remark}}

\newenvironment{proof}[1][Proof]{\noindent\textbf{#1.} }{\ \rule{0.5em}{0.5em}}
\begin{document}

\title{Pricing group membership\thanks{
We thank Kaustav Das, Aditya Goenka and seminar participants at the
University of Birmingham for comments. }}
\author{Siddhartha Bandyopadhyay\thanks{
s.bandyopadhyay@bham.ac.uk }, \\
\relax  Antonio Cabrales\thanks{
a.cabrales@ucl.ac.uk }}
\date{August 2020}
\maketitle

\begin{abstract}
We consider a model where agents differ in their `types' which determines
their voluntary contribution towards a public good. We analyze what the
equilibrium composition of groups are under centralized and centralized
choice. We show that there exists a top-down sorting equilibrium i.e. an
equilibrium where there exists a set of prices which leads to groups that
can be ordered by level of types, with the first k types in the group with
the highest price and so on. This exists both under decentralized and
centralized choosing. We also analyze the model with endogenous group size
and examine under what conditions is top-down sorting socially efficient. We
illustrate when integration (i.e. mixing types so that each group's average
type if the same) is socially better than top-down sorting. Finally, we show
that top down sorting is efficient even when groups compete among themselves.

\noindent \textbf{JEL Classification}: D02, D64, D71, H41

\noindent \textbf{Keywords}: Top down sorting, Group formation, Public good,
Segregation, Integration.
\end{abstract}

\newpage

\section{Introduction}

In this paper we analyze whether groups can `price' group membership to
screen who becomes a member. Joining many groups and clubs is not costless.
Many of them require monetary fees, or non-monetary `sacrifices'\ to become
a member Winslow (1999). Examples abound, new members of a criminal gang or
a terrorist organization have to undergo hazing rituals (Vigil, 1996).
Similarly, new recruits into many military units and college fraternities
also have to go through painful or shameful activities to be accepted into
them (Ostvik and Rudmin 2001, Mercuro, Merritt and Fiumefreddo 2014, De
Klerk 2013, Groves, Griggs and Leflay 2012, Keating et al 2005). Even
religious groups have elaborate initiation rites (Berman 2000, Iannacone
1992). Exclusive clubs, or elite schools require the payment of very
expensive fees to become members (Jenkins, Micklewright and Schnepf 2008).
The question of interest in this paper is the rationale for these practices,
the equilibrium value of such fees, and their welfare implications.

In all those groups, the value of being a member depends on the extent to
which other participants contribute to the common cause or public good they
provide. At the same time, individuals differ in their inclination and
ability to dedicate themselves to the cause, i.e. to contribute to the
public good provided by the group. For this reason, the utility obtained by
any member in belonging to a group will depend on the types of all the other
members. As a consequence, the group will be interested in making sure that
those who gain entry as members are of the right type. But these types are
not necessarily observable, and the entry fee is a way to select members in
an incentive compatible way.

Our model assumes there are a number of participants in a game which has two
stages. The first stage develops as follows. A set of \ `entrepreneurs', one
for each of a fixed number of groups, posts a price for belonging to their
group. The participants then decide, independently and simultaneously,
whether to pay the price to one of those groups, or none at all. In the
second stage, the participants who are in a group decide on their level of
effort/contribution to a public good for the group to which they belong.

The participants differ in the way they benefit from the public good. Each
individual is characterized by a type that is a multiple of the amount of
the public good. This heterogeneity can be variously interpreted as either a
different intrinsic personal enjoyment of the good, or a difference in the
degree of altruism (i.e. some individuals internalize the benefit of the
public good on others to a larger extent). One important implication of this
assumption is that the group members do not care care directly about the
types of others, or about the group size. This is because we assume there is
no direct externality caused by others' types or the number of individuals
in the group. Individuals do care about the actions taken by other group
members, because those actions affect the amount of public good provided
which they enjoy, and contributions to the public good are indeed affected
by their types. This assumption distinguishes our model from others in the
literature of club games (starting with Buchanan, 1965)\ and congestion
games (Oakland 1972, Baumol and Oates 1988).

Clearly, since the groups provide a public good, we would in general have an
underprovision of the public good in the group if the provision is voluntary
and decided on an individual basis. But many of the groups that we have used
to motivate our model have the possibility of imposing a contribution level
within the group. We can think for example of the tightly hierarchical
organization of most armies, and even gangs. For this reason, we also study
group formation when contributions within groups are decided by a planner
within the group.

Whether the individual or the group leader (social planner) decides the
contribution level, our first result is that in both cases we can establish
the existence of a \emph{top-down sorting equilibrium. }This equilibrium has
the following characteristic. The `entrepreneurs' announce a list of
distinct prices, which thus can be ordered from highest to lowest. The
participants sort themselves into groups by their types. A set composed of
the highest types chooses to belong to the group with the highest price.
Another set of types just below the first chooses to belong to the group
with the second highes price. This continues until the last set of types,
who belong to a group with price zero.

We analyze the equilibrium under both centralized and decentralized
decision-making. One difference between the equilibrium under centralized
and decentralized choice of effort is that under decentralized effort choice
every individual gains whenever the average type of the group increases
(although relatively higher types gain more, which is the basis for
segregation). On the other hand, with centralized effort choice within the
group an increase in average type improves the utility of above average
types within the group, but decreases for below average types. This happens
because the group planner does not know the individual types within the
group, and so she requires the same level of contribution from every one.

We next turn our attention to some welfare properties of the decentralized
process for group formation. For this we need to take into account both
group composition and group size. In terms of group composition, we compare
the welfare of the \emph{top-down sorting equilibrium }with the one arising
from an equilibrium in which there is no sorting, and all groups have equal
average types. It turns out that the answer depends on the curvature of the
output function. For example, when the concavity is such the risk aversion
parameter in the CRRA function is bigger than one (i.e. $\alpha >1$), then
the \emph{top-down sorting equilibrium }is less socially beneficial than one
where the population is sorted accross groups so that all groups have the
same average type. We call this `integration'. A corollary of this result is
that in this case, the segregation occurring in equilibrium is socially
inefficient. On the other hand, when $\alpha <1$, then the \emph{top-down
sorting equilibrium }has a larger total welfare than the one achieved when
groups are `integrated' and have the same average types.

In terms of size, for $\alpha >1$, the optimal size of the group is $N\text{,%
}$ the whole population, since individual payoff increases with the size of
the coalition. For $\alpha <1$, the optimal size of the group has to trade
off the benefits from segregation (which are highest if average group
quality increases, say by dropping the lowest types in the group), with
those arising from having larger groups, because group size affects
equilibrium contribution levels.

Finally, up to now we have assumed that there is no inter-group conflict.
However, in some of our applications the groups compete with one another
after they are formed (again, it is easier to think of armies and gangs, but
even educational institutions compete ex post for top-paying jobs). We show
that our \emph{top-down sorting equilibrium }also exists in this context, so
our conclusions are robust to this kind of setting.

Our paper has important implications for policy. There are relevant
environments ($\alpha >1$)\ where the social planner would like an
`integrated' society, which might not arise in equilibrium. She wants to
achieve this purely on social surplus maximization grounds. This provides a
new rationale for integration in various social domains, like education and
housing, without having to resort to preference about equality. We will
discuss this in more depth in the conclusion.

\subsection{Literature}

Our paper develops a novel theory of group formation. Clearly, this paper
links to the classical literature on club theory (Buchanan 1965, Berglas
1976). This literature does not consider differences in information about
preferences, and therefore it does not provide a rationale for undertaking
actions to signal preferences. Somewhat closer to our work is Ben Porath and
Dekel (1992), who use the potential for self sacrifice as a way to signal
future intentions. However, they do it because of equilibrium selection
problems, and the self sacrifice does not actually occur in equilibrium.
Helsley and Strange (1991), in turn, model club formation in a context with
homogeneous tastes and costs. In their model fees and prices are used to
obtain second best usage when congestion within a club is a problem.

In terms of theory, two close papers to our's are Jaramillo and Moizeau
(2002), and Cornes and Silva (2013). Jaramillo and Moizeau (2003) study a
model where individuals differ in income, and higher income people desire a
different level of the local public good. Since information about income is
private, individuals use costly signalling to join a club with others who
have similar income (and, hence, preferences for the local public good).
Different from us, the group formation is uncoordinated (there is no
entrepreneur creating the clubs), there are only two types, and they do not
contemplate the possibility of internal coordination of contributions. Also,
in their model there is no reason for individuals of different types to
group together in the social optimum. In Cornes and Silva (2013) some people
love prestige (they have higher utility from contributing relative to the
average) and others are purists. Fees serve to sort them into clubs. The
model, unlike ours has both positive and negative externalities, but the
competitive equilibrium is inefficient. Competition yields co-ordination
benefits with the formation of prestige clubs.

Our paper also has some connection to a literature that is related to group
formation in contexts were types are differentiated horizontally. For
example, Levy and Razin (2012) analyze explicit displays of religious
beliefs and cooperation within the religious group. Baccara and Yariv
(2013), also analyze group formation and contribution to horizontally
differentiated tasks.

Another relevant literature for us relates to assortative matching. Durlauf
and Seshadri (2003) examines when assortative matching is efficient. Legros
and Newman (2007) study sufficient conditions for assortative matching in
equilibrium. Hoppe, Moldovanu and Sela (2008) analyze when costly signaling
is necessary for assortative matching under incomplete information.
Importantly, they check when gains form assortative matching are offset by
signalling costs. There is also a large literature studying sorting into
schools, to take advantage of peer effects (Epple Romano 1998, Cullen,
Jacob, and Levitt 2003, Hsieh and Urquiola 2006 or MacLeod and Urquiola
2015).

There is a body of \ work in the experimental literature exploring how
endogenous sorting into groups can help solve the problem of free riding.
For example, costly rituals have been shown to promote greater co-operation
(Sosis 2004, Ruffle and Sosis 2007). Page, Putterman and Unel (2005) show
that endogenous segregation helps avoid free-riding and provides some
support for the efficiency of top down sorting. Cimino (2011) does a survey
based experiment where participants are asked about an initiation activity
for a group, which itself provides benefits to members. Those participants
randomly allocated to a group providing a good with a more public component
were significantly more likely to choose a stressful initiation task. This
sort of hazing can be considered as an entry fee to prevent free riding.
This would be very much in the spirit of our model. Aimone et al. (2013)
show that something akin to our endogenous segregation equilibrium happens
in the lab. In their experiment participants can choose to participate in a
voluntary contribution public good game in groups having different rates of
return. People that are more prone to contribute to public goods join those
groups with lower rates of return, thus signalling their type. As in our
model, a costly choice (in this case, choosing a less efficient technology)
is a credible signal of a type wishing to contribute more to the public good.

In addition to experimental results, there is a significant ethnographic
evidence that is connected to our issues. For example, Vigil (1996) analyzes
gang initiation. This study is consistent with our model, i.e. initiation
rites are used to screen potential members. Sosis, Kress and Boster (2007)
show the importance of such costly male rituals in signalling commitment and
promoting solidarity among men who need to organise themselves for warfare.
Soler (2012) rationalises the existence of wasteful religious rituals. It
analyzes the practice of Candomble, an Afro Brazilian ritual. It shows that
participation is correlated with higher contributions to public goods, which
is consistent with our model. Cleaver (2004) shows that entry fees are used
in clubs to exclude the ``non elite'' .

\section{Model: signaling and coalition formation}

This model has two stages. There are a large and finite number of agents
playing the game $\mathcal{N}$. Each agent $j \in \mathcal{N}$ is
characterized by a level of preference for the public good $\xi _{j}\text{.}$
In the first stage the agents form coalitions of a fixed number of players $N%
\text{,}$ in a way we will describe in section \ref{hazing section}. In the
second stage, once a coalition is formed, every agent simultaneously decides
how much to contribute towards a public good, $x_{j}$ taking as given the
coalition structure.

The total amount of the public good is given by $V \left (\sum _{i
=1}^{N}x_{i}\right )\text{,}$ with $V \left (\text{.}\right )$ being a
strictly concave twice continuously differentiable function. The personal
cost of the contribution $x_{j}$ is given by $x_{j}^{2}/2.$ The utility of
every agent after the coalition is formed is 
\begin{equation*}
U_{j} =\xi _{j} V \left (\sum _{i =1}^{N}x_{i}\right ) -\frac{1}{2} x_{j}^{2}
\end{equation*}%
The parameter $\xi _{j}$ indicates that the individual may care for more
than only his own utility but she internalizes the benefits of other
players. A player with a $\xi _{j} >1$ is thought to be ``altruistic''\
whereas $\xi _{j} =1$ is a selfish player. The value of $\xi _{j}$ is
obviously relevant for all players in any coalition, as it increases the
marginal value of contributions of their owners. It is also the private
information of the players, which can be (partially) solved using a pre-game
costly signaling exercise which we now describe.

\subsection{The signaling game\label{hazing section}}

Consider a finite set of coalitions $l \in \{1 ,\ldots ,L\}$, each with $N$
slots.\footnote{%
We assume for simplicity that all coalitions are equally sized. Nothing
substantial changes if they are still exogenous but differently sized. Later
on we address endogenously sized coalitions.} Let $\bar{\xi }_{l} =\sum _{i
\in l}\xi _{i}\text{.}$ Assume as well that each $N$ is large enough so that
the compositional impact of changing one member's type on the $\bar{\xi }%
_{l} $ of coalition $l$ is negligible. The coalition formation game is as
follows. A set of `entrepreneurs'\ post a set of prices $p_{l}$ for $l \in
\{1 ,\ldots ,L\}\text{.}$ We think of prices that need not be `monetary'.
Any costly action whose effect in utility is separable would work. Then all
individuals in the population decide which coalition to join (and thus pay
the price of joining). If a coalition is oversubscribed (it has $M >N$
candidates applying to it), then a fair lottery decides which $M -N$
candidates get allocated to not being in a coalition, which is free and
obviously has enough capacity for the full group.

Order arbitrarily the available coalitions. We denote by \textsl{top-down
sorting} the following assignment of members into coalitions according to
their type. Coalition $1$ gets assigned the $N$ highest type members,
coalition $2$ the $N$ highest type members among the remaining ones, and so
on until all members are assigned to one (and only one) coalition. The
top-down sorting leads to a coalition structure with types stratified from
higher to lower. Namely, given two coalitions $l >k$ and two members $i ,j$
that are assigned to either coalition by top-down sorting, then, $\xi _{i}
>\xi _{j}~$and $\bar{\xi }_{l} \geq \bar{\xi }_{k}$. To ensure that this
inequality is strict for at least one pair of players in two different
coalitions, we assume that two successive coalitions cannot be fully
occupied by players of the same type. As mentioned above to join a coalition 
$l$, they must choose an action with cost $p_{l}$. The last group is the
option not being in a coalition. We say that an assignment of members to
coalitions and a vector of entry costs forms an equilibrium when, given the
costs, no individual prefers to change coalitions and either a coalition is
full or its associated cost is zero.

Denote by 
\begin{equation*}
\frac{ \partial U_{j}}{ \partial \bar{\xi }_{l}}
\end{equation*}%
the derivative of the equilibrium utility of individual $j$ with respect to
changes in the average type $\bar{\xi }_{l}$ of the coalition to which they
belong. The differential sensitivity of different types of potential
coalition members to the composition of coalitions has implications for
coalition formation that we now analyze.

\begin{proposition}
\label{segregate}There exists an assignment equilibrium with \textsl{%
top-down sorting} if whenever $i$ and $j$ are such that $\xi _{i} >\xi _{j}$
we have that 
\begin{equation}
\frac{ \partial U_{i}}{ \partial \bar{\xi }_{l}} -\frac{ \partial U_{j}}{
\partial \bar{\xi }_{l}} >0  \label{part a}
\end{equation}%
Letting $\xi _{i^{ \ast } (l)}$ be the type of the lowest member in
coalition $l$, the fee for a full coalition $l$ is defined recursively as: 
\begin{equation}
p_{l} =\int \limits _{\bar{\xi }_{l +1}}^{\bar{\xi }_{l}}\frac{ \partial
U_{_{i^{ \ast } (l)}}}{ \partial \bar{\xi }_{j}} d \bar{\xi }_{j} +p_{l +1}
,l =1 ,\ldots ,L -1 ,  \label{fee}
\end{equation}%
with $p_{L} =0$.
\end{proposition}

\begin{proof}
See Appendix.
\end{proof}

This condition provides a test for the existence of endogenous sorting.

\subsection{Trust and group contributions}

Given the coalitions that have formed, we now analyze the equilbria in the
subgames where the coalitions are present. Remember that 
\begin{equation*}
U_{j} =\xi _{j} V \left (\sum _{i =1}^{N}x_{i}\right ) -\frac{1}{2} x_{j}^{2}
\end{equation*}%
The FOC of the second stage problem are 
\begin{equation*}
\frac{ \partial U_{j}}{ \partial x_{j}} =\xi _{j} V^{ \prime } \left (\sum
_{i =1}^{N}x_{i}\right ) -x_{j} =0
\end{equation*}%
This implies that for all $i ,j$ 
\begin{equation*}
\frac{x_{j}}{\xi _{j}} =\frac{x_{i}}{\xi _{i}}
\end{equation*}%
So normalizing $\xi _{1} =1$ 
\begin{equation*}
\xi _{i} x_{1} =x_{i}
\end{equation*}%
and hence in equilibrium 
\begin{equation}
V^{ \prime } \left (x_{1} \sum _{i =1}^{N}\xi _{i}\right ) =x_{1}
\label{FOC x1}
\end{equation}%
which yields a unique equilibrium given the composition of the coalition.
Let a particular coalition $A$ composed of a group of people with qualities $%
\left (\xi _{1} ,\xi _{2} ,\ldots ,\xi _{n}\right )$ with 
\begin{equation*}
\bar{\xi }_{l} =\frac{1}{N} \sum _{i =1}^{N}\xi _{i}
\end{equation*}%
Totally differentiating \ref{FOC x1} with respect to $\bar{\xi }_{l}$ we get;

\begin{equation*}
V^{ \prime \prime } \left (N x_{1} \bar{\xi }_{l}\right ) \frac{ \partial
x_{1}}{ \partial \bar{\xi }_{l}} N \bar{\xi }_{l} +V^{ \prime \prime } \left
(N x_{1} \bar{\xi }_{l}\right ) N x_{1} =\frac{ \partial x_{1}}{ \partial 
\bar{\xi }_{l}}
\end{equation*}
\begin{equation*}
\frac{ \partial x_{1}}{ \partial \bar{\xi }_{l}} =\frac{V^{ \prime \prime }
\left (N x_{1} \bar{\xi }_{l}\right ) N x_{1}}{1 -V^{ \prime \prime } \left
(N x_{1} \bar{\xi }_{l}\right ) N \bar{\xi }_{l}} <0
\end{equation*}%
so contributions decrease with average quality. Nevertheless equilibrium
utility

\begin{equation*}
U_{j} =\xi _{j} V \left (N x_{1} \bar{\xi }_{l}\right ) -\frac{\xi _{j}^{2}}{%
2} x_{1}^{2}
\end{equation*}%
so that 
\begin{align*}
\frac{ \partial U_{j}}{ \partial \bar{\xi }_{l}} & = & \xi _{j} V^{ \prime }
\left (N x_{1} \bar{\xi }_{l}\right ) N \bar{\xi }_{l} \frac{ \partial x_{1}%
}{ \partial \bar{\xi }_{l}} +\xi _{j} N x_{1} V^{ \prime } \left (N x_{1} 
\bar{\xi }_{l}\right ) -\xi _{j}^{2} x_{1} \frac{ \partial x_{1}}{ \partial 
\bar{\xi }_{l}} \\
& = & \xi _{j} N x_{1}^{2} \binom{1 -V^{ \prime \prime } \left (N x_{1} \bar{%
\xi }_{l}\right ) \xi _{j}}{1 -V^{ \prime \prime } \left (N x_{1} \bar{\xi }%
_{l}\right ) N \bar{\xi }_{l}} >0
\end{align*}

Furthermore, from inspection it is clear that if $\bar{\xi }_{l}$ changes
little with $\xi _{j}$, say because every individual is negligible, then 
\begin{equation}
\frac{ \partial ^{2}U_{j}}{ \partial \bar{\xi }_{l} \partial \xi _{j}} >0
\label{sorting decentralized}
\end{equation}

From equation (\ref{sorting decentralized}) it is immediate that

\begin{proposition}
\label{segregate decentralized}There exists an assignment equilibrium with 
\textsl{top-down sorting} under decentralized effort choice.
\end{proposition}

\subsection{The public good game: centralized choosing}

Suppose instead that the level of contributions in a coalition are decided
centrally by a utilitarian social planner.

\begin{equation*}
U_{j} =\xi _{j} V \left (\sum _{i =1}^{N}x_{i}\right ) -\frac{1}{2} x_{j}^{2}
\end{equation*}

\begin{equation*}
\mathcal{U}= \sum_{j=1}^{N}U_{j}=\sum_{j=1}^{N}\left( \xi _{j}V\left(
\sum_{i=1}^{N}x_{i}\right) -\frac{1}{2}x_{j}^{2}\right) =N\bar{\xi}%
_{l}V\left( \sum_{i=1}^{N}x_{i}\right) -\frac{1}{2}\sum_{j=1}^{N}x_{j}^{2}
\end{equation*}%
The FOC of the second stage problem are now 
\begin{equation*}
\frac{\partial \mathcal{U}}{\partial x_{j}}=N\bar{\xi}_{l}V^{\prime }\left(
\sum_{i=1}^{N}x_{i}\right) -x_{j}=0
\end{equation*}%
so for all $i,j$ within a coalition 
\begin{equation*}
x_{i}=x_{j}
\end{equation*}%
\begin{equation}
N\bar{\xi}_{l}V^{\prime }\left( Nx_{1}\right) -x_{1}=0  \label{FOC x2}
\end{equation}

Totally differentiating \ref{FOC x2} with respect to $\bar{\xi }_{l}$ we get;

\begin{equation*}
N \frac{ \partial x_{1}}{ \partial \bar{\xi }_{l}} V^{ \prime \prime } \left
(N x_{1}\right ) N \bar{\xi }_{l} +N V^{ \prime } \left (N x_{1}\right ) =%
\frac{ \partial x_{1}}{ \partial \bar{\xi }_{l}}
\end{equation*}
\begin{equation*}
\frac{ \partial x_{1}}{ \partial \bar{\xi }_{l}} =\frac{V^{ \prime } \left
(N x_{1} \bar{\xi }_{l}\right ) N}{1 -V^{ \prime \prime } \left (N
x_{1}\right ) N^{2} \bar{\xi }_{l}} >0
\end{equation*}%
so contributions increase with average quality. Hence equilibrium utility

\begin{equation*}
U_{j} =\xi _{j} V \left (N x_{1}\right ) -\frac{x_{j}^{2}}{2}
\end{equation*}%
remember that $N \bar{\xi }_{l} V^{ \prime } \left (N x_{1}\right ) -x_{1} =0%
\text{,}$ so that 
\begin{align*}
\frac{ \partial U_{j}}{ \partial \bar{\xi }_{l}} & = & \xi _{j} V^{ \prime }
\left (N x_{1}\right ) N \frac{ \partial x_{1}}{ \partial \bar{\xi }_{l}}
-x_{1} \frac{ \partial x_{1}}{ \partial \bar{\xi }_{l}} \\
& = & \xi _{j} \frac{x_{1}}{\bar{\xi }_{l}} \frac{ \partial x_{1}}{ \partial 
\bar{\xi }_{l}} -x_{1} \frac{ \partial x_{1}}{ \partial \bar{\xi }_{l}} =%
\binom{\xi _{j} -\bar{\xi }_{l}}{\bar{\xi }_{l}} x_{1} \frac{ \partial x_{1}%
}{ \partial \bar{\xi }_{l}} \\
& = & \binom{\xi _{j} -\bar{\xi }_{l}}{\bar{\xi }_{l}} x_{1} \frac{V^{
\prime } \left (N x_{1} \bar{\xi }_{l}\right ) N}{1 -V^{ \prime \prime }
\left (N x_{1}\right ) N^{2} \bar{\xi }_{l}}
\end{align*}%
where the sign is positive for $j$ with $\xi _{j} -\bar{\xi }_{l} >0$ and
negative otherwise.

\begin{remark}
It is clear that with respect to the decentralized equilibrium some types of
players, i.e. those with a higher than average type within a coalition, have
a higher utility while others, i.e. those with lower than average type, have
a lower utility.
\end{remark}

Furthermore, from inspection it is clear that if changing a single $\xi _{j}$
does not change much the average type of a coalition, then 
\begin{equation}
\frac{ \partial ^{2}U_{j}}{ \partial \bar{\xi }_{l} \partial \xi _{j}} >0
\label{sorting centralized}
\end{equation}

From equation (\ref{sorting centralized}) it is immediate that

\begin{proposition}
\label{segregate centralized}There exists an assignment equilibrium with 
\textsl{top-down sorting} under centralized effort choice.
\end{proposition}

\section{Endogenous size and social optima}

In the previous sections, the size of coalitions has been exogenously fixed
at $N\text{.}$ But given the environment considered, it would natural to
consider the equilibrium when the coalition size is also endogenous. This
could have important implications both for the equilibrium contributions,
and for the efficient composition and size of the groups. One main tradeoff
is the following. In a larger group, the free-riding \ can become more
problematic. On the other hand, in a larger group, the marginal benefits of
an action positively affect a larger set of people. What is optimal will
likely depend on specific features of the technology.

In this section we show that when the group size can be centrally chosen a
utilitarian social planner would sometimes prefer top-down sorting (we call
this `segregation') while under other conditions she would prefer types to
mix so that all groups have the same average type (we call this
`integration').

We now use a CRRA $V \left (\text{.}\right )$ function to analyze this
problem. 
\begin{equation*}
V \left (\sum _{i =1}^{N}x_{i}\right ) =\frac{\left (\sum _{i
=1}^{N}x_{i}\right )^{1 -a} -1}{1 -a}
\end{equation*}

In order for this $V \left (\text{{\textperiodcentered}}\right )$ function
to make sense as a production function, we assume $x_{i} \geq 1.$ Note that
there is no risk in this problem, so we use the CRRA function as a
convenient way to parameterize concavity.

Then, for a given $N$ the FOC for coalition efforts in this case are

\begin{align*}
N\bar{\xi}_{l}\left( \left( Nx_{1}\right) ^{-a}\right) -x_{1}& = & & 0 \\
\bar{\xi}_{l}N^{1-a}& = & & x_{1}^{1+a}
\end{align*}

Remember that in coalition all members choose the same value under
centralized decision-making and so $x_{i} =x_{j} =x_{1}$ within the
coalition. Thus the total utility within the coalition is 
\begin{equation*}
\mathcal{U}_{l} =N \bar{\xi }_{l} \frac{\left (N \left (\bar{\xi }_{l} N^{1
-a}\right )^{\frac{1}{1 +a}}\right )^{1 -a} -1}{1 -a} -\frac{1}{2} N \left
(\left (\bar{\xi }_{l} N^{1 -a}\right )^{\frac{1}{1 +a}}\right )^{2}
\end{equation*}%
so that 
\begin{equation*}
\mathcal{U}_{l} =\left (\frac{1}{1 -a} -\frac{1}{2}\right ) \bar{\xi }_{l}^{%
\frac{2}{1 +a}} N^{\frac{3 -a}{1 +a}} -\frac{N \bar{\xi }_{l}}{1 -a}
\end{equation*}%
and society welfare is 
\begin{equation}
\mathcal{U} =\sum _{l =1}^{L}\mathcal{U}_{l} =\sum _{l =1}^{L}\left (\left (%
\frac{1}{1 -a} -\frac{1}{2}\right ) \bar{\xi }_{l}^{\frac{2}{1 +a}} N^{\frac{%
3 -a}{1 +a}} -\frac{N \bar{\xi }_{l}}{1 -a}\right )
\label{endogenous segregation condition}
\end{equation}

Define $\mathcal{U}_{T D}$ as the total utility obtained in society when
coalitions are formed with \textsl{top-down sorting} and $V \left (\sum _{i
=1}^{N}x_{i}\right ) =\left (\left (\sum _{i =1}^{N}x_{i}\right )^{1 -a}
-1\right )/\left (1 -\alpha \right )\text{.}$

Suppose that it is possible to organize coalitions so that all coalitions
have the same mean $\bar{\xi }_{l}$ and define $\mathcal{U}_{I}$ as the
total utility obtained in society when all coalitions have the same mean $%
\bar{\xi }_{l}^{ \ast } =\sum _{l =1}^{L}\bar{\xi }_{l}/L$ where $\bar{\xi }%
_{l}$ is the group $l$ mean under top-down sorting. Because all the groups
have the same average type, the optimally chosen contributions are the same
\ in all of them and thus no one has an incentive to choose the group to
which they belong.

Remember that the central planner does not know the types of any of the
players. Thus, we assume groups are sufficiently large that simply
allocating individuals randomly to groups generates groups with equal
average types in expectation.

\begin{proposition}
\label{propsegregation}Suppose coalitions size is exogenously set to $N\text{%
.}$ Then $\mathcal{U}_{T D} >\mathcal{U}_{I}$ if and only if $a <1.$
\end{proposition}

\begin{proof}
This follows from (\ref{endogenous segregation condition}) since

\begin{align*}
\mathcal{U}_{T D} & = & \sum _{l =1}^{L}\left (\left (\frac{1}{1 -a} -\frac{1%
}{2}\right ) \bar{\xi }_{l}^{ \ast \frac{2}{1 +a}} N^{\frac{3 -a}{1 +a}} -%
\frac{N \bar{\xi }_{l}^{ \ast }}{1 -a}\right ) =\left (\frac{1}{1 -a} -\frac{%
1}{2}\right ) N^{\frac{3 -a}{1 +a}} L \bar{\xi }_{l}^{ \ast \frac{2}{1 +a}} -%
\frac{L N \bar{\xi }_{l}^{ \ast }}{1 -a} \\
& = & \left (\frac{1}{1 -a} -\frac{1}{2}\right ) N^{\frac{3 -a}{1 +a}} L 
\binom{\sum _{l =1}^{L}\bar{\xi }_{l}}{L}^{\frac{2}{1 +a}} -\frac{N \sum _{l
=1}^{L}\bar{\xi }_{l}}{1 -a} \\
\mathcal{U}_{I} & = & \sum _{l =1}^{L}\left (\left (\frac{1}{1 -a} -\frac{1}{%
2}\right ) \bar{\xi }_{l}^{\frac{2}{1 +a}} N^{\frac{3 -a}{1 +a}} -\frac{N 
\bar{\xi }_{l}}{1 -a}\right ) =\left (\frac{1}{1 -a} -\frac{1}{2}\right ) N^{%
\frac{3 -a}{1 +a}} \sum _{l =1}^{L}\bar{\xi }_{l}^{\frac{2}{1 +a}} -\frac{N
\sum _{l =1}^{L}\bar{\xi }_{l}}{1 -a}
\end{align*}%
so that by applying Jensen's inequality, we have that 
\begin{equation*}
\mathcal{U}_{T D} >\mathcal{U}_{I}
\end{equation*}
\end{proof}

\begin{equation}
\frac{\mathcal{U}_{l}}{N}=\frac{1+a}{2\left( 1-a\right) }\bar{\xi}_{l}^{%
\frac{2}{1+a}}N^{\frac{2-2a}{1+a}}-\frac{\bar{\xi}_{l}}{1-a}
\label{average payoff}
\end{equation}%
\begin{equation*}
\frac{\partial \left( \frac{\mathcal{U}_{l}}{N}\right) }{\partial N}=\bar{\xi%
}_{l}^{\frac{2}{1+a}}N^{\frac{2-2a}{1+a}}>0
\end{equation*}%
This implies that for a given $\bar{\xi}_{l}$ the average payoff in a group
increases in $N$ independently of the concavity of the individual functional
form within the CRRA class.

Given that proposition \ref{propsegregation} establishes that average payoff
in the coalition is concave in average type for $a >1$ and thus you want to
form heterogeneous coalitions, it follows immediately that

\begin{proposition}
For $a >1$ the social planner would like a single group of the maximal size.
\end{proposition}

When, on the other hand $a <1\text{,}$ there is a tradeoff for the ``high''\
quality groups. On the one hand, they would prefer to have total segregation
to increase payoff, since average payoff increases with average type as in
equation \ref{average payoff} we see that it increases in $\bar{\xi }_{l}^{%
\frac{2}{1 +a}}$. On the other hand, a bigger size increases payoff within
the group, as in equation \ref{average payoff} we see that it increases in $%
N^{\frac{2 -2 a}{1 +a}}$. Clearly this tradeoff pushes for at least some
degree of mixing.

\section{Inter-coalition competition}

So far, we have studied the problem of coalition formation and activities as
if the activities of those coalitions did not interact with one another. But
in our motivation we discussed the evidence that group formation often
occurs in contexts where the groups compete, such as gangs (Vigil 1996) or
warfare (Sosis, Kress and Boster 2007). For this reason we will now study
the case where, after the groups form, they compete. The main insight of
this section, is that our earlier results also apply in this case. In
particular, let us assume that the coalitions compete, after having formed,
in a contest. We will show that the incentives for coalition formation that
allow for a top-down sorting assignment (as in Proposition \ref{segregate
decentralized}) still hold in this case. To be more precise, take two
coalitions with respective sizes $M$ and $N$ but otherwise identical utility
functions.

Then assume that the payoffs are given by\footnote{%
We assume that payoffs are the outcomes of a contest, and given by a contest
success function, as in Skaperdas (1996).} 
\begin{equation*}
U_{j} =\frac{\xi _{j}}{N} \frac{V \left (\sum _{i =1}^{N}x_{i}\right )}{V
\left (\sum _{i =1}^{N}x_{i}\right ) +V \left (\sum _{i =1}^{M}y_{i}\right )}
-\frac{1}{2} x_{j}^{2}
\end{equation*}%
and 
\begin{equation*}
V \left (\sum _{i =1}^{N}x_{i}\right ) =\left (\sum _{i =1}^{N}x_{i}\right
)^{b} ;V \left (\sum _{i =1}^{M}y_{i}\right ) =\left (\sum _{i
=1}^{M}y_{i}\right )^{b}
\end{equation*}

\bigskip and $b <1$ which as shown in Proposition \ref{propsegregation} is
the case favorable to segregation into more than one coalition.

\begin{proposition}
\label{segregate decentralized competition}There exists an assignment
equilibrium with \textsl{top-down sorting} under decentralized effort choice
when coalitions compete.
\end{proposition}

\begin{proof}
See Appendix.
\end{proof}

\section{Conclusion}

We have studied a model of public good provision within groups. The group
members have heterogeneous preferences for the public good, and high types
contribute more towards its provision. Thus, all individuals prefer to be in
groups with higher average types. We allow for the existence of
`entrepreneurs' who create the groups and demand a (possibly non-monetary)
`fee' to enter the group. We study the case where the `entrepreneurs' can
enforce a contribution level within the group and also when once in the
group, contributions are voluntary.

We first show that both under centralized and decentralized provision there
is \emph{top-down sorting equilibrium }in which groups are organized
assortatively by type. Under centralized contributions every participant is
better off if average types increase. When contributions are centralized
that is not true and only above average types are better off.

There are interesting results in terms of welfare. Under more concave
utility functions, utilitarian welfare is higher when average types are the
same across groups ( `integration') than under \emph{top-down sorting }(
`segregation'). With less concave utilities, the opposite is true. Concavity
also favors large groups. Under less concave functions there is a trade-off,
as very homogeneous groups are good, but large groups are also good, so one
could be willing to sacrifice in terms of homogeneity to increase
contributions because of group size. The results are robust to environments
in which groups compete among themselves.

While our results provide important insights, there are of course
limitations to our analysis. For example, we have not studied repeated
interactions, which, through reciprocity, could lead to different results.
We conjecture that in a repeated environment we would be more likely to
observe the results we posit under `centralized' group management. We think
the reasons for assortative matching will survive in that case.

We have studied an environment where all the agents are `symmetric'\ even if
they have different types. However, there are important applications where
matching is bilateral, like the marriage market. We think that some of our
insights will carry over in those applications. Indeed\ Greenwood et al.
(2014) document an increase in marriage market assortativity as an important
source of inequality. However, it is not immediately obvious that our
welfare results on sorting will carry over in that context. Another
intriguing question in that environment has to do with the fact that in many
species, it is only one of the sides of the marriage market that engages in
costly signaling before mating (Jiang, Bolnick and Kirkpatrick 2014).

Our paper also provides interesting insights for public policy. Most extant
foundations for \ `integration' policies, either in housing or school
settings have to do either with concerns to reduce inequality (Ananat 2011,
Reardon 2016) or with direct spillovers from higher ability individual on
other individuals (Duflo, Dupas, Kremer 2011, Graham, Imbens, Ridder 2014).
We provide a new foundation for integration that is based on indirect
spillovers. Individuals do not care directly about the type of others', they
care because high types provide higher effort towards public good provision.

It is worth commenting that the integration result in our model comes from
the objective of social surplus maximization rather than equity
considerations. Nonetheless, it has equity implications as low types benefit
from the higher provision of public goods in their group. Thus, the
efficient distribution is also equitable and may provide a new rationale for
integration efforts in different societies. Our model may also provide a
different rationale to whether the attainment gap of children (of lower
socio-economic status) in segregated neighborhoods (see for instance Ananat,
2011) may have something to do with the lower provision of public good in
such neighborhoods. If so, this may have positive implications, particularly
across generations, which may further strengthen case for integration. We
leave this interesting issue for future work.

\textbf{{\LARGE Appendix A\label{Appendix}}} \setcounter{section}{0}\bigskip

\renewcommand{\thesection}{\Alph{section}}

\textbf{Proof of Proposition \ref{segregate}}

A member of coalition $l$ with type $\xi _{i}$ does not want to move to
coalition $l +1$ provided that: 
\begin{align*}
U_{i} \left (\bar{\xi }_{l}\right ) -p_{l} & \geq & U_{i} \left (\bar{\xi }%
_{l +1}\right ) -p_{l +1} \\
U_{i} \left (\bar{\xi }_{l}\right ) -U_{i} \left (\bar{\xi }_{l +1}\right )
& \geq & p_{l} -p_{l +1}\text{.}
\end{align*}%
Such person will have a type such that $\xi _{i^{ \ast } (l -1)} \geq \xi
_{i} \geq $ $\xi _{i^{ \ast } (l)}\text{.}$ Then we have that:

\begin{align*}
U_{i} \left (\bar{\xi }_{l}\right ) -U_{i} \left (\bar{\xi }_{l +1}\right )
& = & \int \limits _{\bar{\xi }_{l +1}}^{\bar{\xi }_{l}}\frac{ \partial U_{i}%
}{ \partial \bar{\xi }_{j}} d \bar{\xi }_{j} \\
& \geq & \int \limits _{\bar{\xi }_{l +1}}^{\bar{\xi }_{l}}\frac{ \partial
U_{i^{ \ast } (l)}}{ \partial \bar{\xi }_{j}} d \bar{\xi }_{j} \\
& = & p_{l} -p_{l +1}\text{,}
\end{align*}%
where the inequality is true by (\ref{part a}). Similarly a member of
coalition $l$ with type $\xi _{i}$ does not want to move to coalition $l -1$
provided that: 
\begin{align*}
U_{i} \left (\bar{\xi }_{l}\right ) -p_{l} & \geq & U_{i} \left (\bar{\xi }%
_{l -1}\right ) -p_{l -1} \\
p_{l -1} -p_{l} & \geq & U_{i} \left (\bar{\xi }_{l -1}\right ) -U_{i} \left
(\bar{\xi }_{l}\right )
\end{align*}%
Remember that $\xi _{i^{ \ast } (l -1)} \geq \xi _{i} \geq $ $\xi _{i^{ \ast
} (l)}\text{.}$ Thus: 
\begin{align*}
p_{l -1} -p_{l} & = & \int \limits _{\bar{\xi }_{l}}^{\bar{\xi }_{l -1}}%
\frac{ \partial U_{i^{ \ast } (l -1)}}{ \partial \bar{\xi }_{j}} d \bar{\xi }%
_{j} \\
& \geq & \int \limits _{\bar{\xi }_{l}}^{\bar{\xi }_{l -1}}\frac{ \partial
U_{i}}{ \partial \bar{\xi }_{j}} d \bar{\xi }_{j} \\
& = & U_{i} \left (\bar{\xi }_{l -1}\right ) -U_{i} \left (\bar{\xi }%
_{l}\right )\text{,}
\end{align*}%
where, again, the inequality is true by (\ref{part a}). $\blacksquare $
\newpage

\textbf{Proof of Proposition \ref{segregate decentralized competition}}

The FOC for the different players are now 
\begin{equation*}
\frac{\xi _{j}}{N}\frac{b\left( \sum_{i=1}^{M}y_{i}\right) ^{b}\left(
\sum_{i=1}^{N}x_{i}\right) ^{b-1}}{\left( \left( \sum_{i=1}^{N}x_{i}\right)
^{b}+\left( \sum_{i=1}^{M}y_{i}\right) ^{b}\right) ^{2}}=x_{j}
\end{equation*}%
\begin{equation*}
\frac{1}{N}\frac{b\left( \sum_{i=1}^{M}y_{i}\right) ^{b}\left(
\sum_{i=1}^{N}x_{i}\right) ^{b-1}}{\left( \left( \sum_{i=1}^{N}x_{i}\right)
^{b}+\left( \sum_{i=1}^{M}y_{i}\right) ^{b}\right) ^{2}}=\frac{x_{j}}{\xi
_{j}}=x_{1}
\end{equation*}%
\begin{equation*}
\frac{1}{N}\frac{b\left( y_{1}\sum_{i=1}^{M}\xi _{i}\right) ^{b}\left(
x_{1}\sum_{i=1}^{N}\xi _{i}\right) ^{b-1}}{\left( \left(
x_{1}\sum_{i=1}^{N}\xi _{i}\right) ^{b}+\left( y_{1}\sum_{i=1}^{M}\xi
_{i}\right) ^{b}\right) ^{2}}=x_{1}
\end{equation*}%
\begin{equation}
\frac{1}{N}\frac{b\left( y_{1}M\bar{\xi}_{k}\right) ^{b}\left( x_{1}N\bar{\xi%
}_{l}\right) ^{b-1}}{\left( \left( x_{1}N\bar{\xi}_{l}\right) ^{b}+\left(
y_{1}M\bar{\xi}_{k}\right) ^{b}\right) ^{2}}=x_{1}  \label{comp1}
\end{equation}%
\begin{equation*}
\frac{1}{M}\frac{b\left( y_{1}M\bar{\xi}_{k}\right) ^{b-1}\left( x_{1}N\bar{%
\xi}_{l}\right) ^{b}}{\left( \left( x_{1}N\bar{\xi}_{l}\right) ^{b}+\left(
y_{1}M\bar{\xi}_{k}\right) ^{b}\right) ^{2}}=y_{1}
\end{equation*}%
\begin{align*}
\frac{\bar{\xi}_{k}M^{2}y_{1}}{\bar{\xi}_{l}N^{2}x_{1}}& = & & \frac{x_{1}}{%
y_{1}} \\
\sqrt{\bar{\xi}_{k}}My_{1}& = & & \sqrt{\bar{\xi}_{l}}Nx_{1}
\end{align*}%
\begin{equation}
\bar{\xi}_{k}My_{1}=\sqrt{\bar{\xi}_{k}}\sqrt{\bar{\xi}_{l}}Nx_{1}
\label{comp2}
\end{equation}%
and substituting (\ref{comp2}) into (\ref{comp1}) we get 
\begin{equation*}
x_{1}=\frac{1}{N}\frac{b\left( y_{1}M\bar{\xi}_{k}\right) ^{b}\left( x_{1}N%
\bar{\xi}_{l}\right) ^{b-1}}{\left( \left( x_{1}N\bar{\xi}_{l}\right)
^{b}+\left( y_{1}M\bar{\xi}_{k}\right) ^{b}\right) ^{2}}=\frac{1}{x_{1}}%
\frac{1}{N}\frac{b\left( \sqrt{\bar{\xi}_{k}}\sqrt{\bar{\xi}_{l}}N\right)
^{b}\left( N\bar{\xi}_{l}\right) ^{b-1}}{\left( \left( N\bar{\xi}_{l}\right)
^{b}+\left( \sqrt{\bar{\xi}_{k}}\sqrt{\bar{\xi}_{l}}N\right) ^{b}\right) ^{2}%
}
\end{equation*}%
\begin{equation*}
x_{1}^{2}=\frac{1}{N^{2}}\frac{b\left( \sqrt{\bar{\xi}_{k}}\sqrt{\bar{\xi}%
_{l}}\right) ^{b}\left( \bar{\xi}_{l}\right) ^{b-1}}{\left( \left( \bar{\xi}%
_{l}\right) ^{b}+\left( \sqrt{\bar{\xi}_{k}}\sqrt{\bar{\xi}_{l}}\right)
^{b}\right) ^{2}}=\frac{1}{N^{2}}\frac{1}{\bar{\xi}_{l}}\frac{b\left( \sqrt{%
\bar{\xi}_{k}}\sqrt{\bar{\xi}_{l}}\right) ^{b}}{\left( \left( \sqrt{\bar{\xi}%
_{l}}\right) ^{b}+\left( \sqrt{\bar{\xi}_{k}}\right) ^{b}\right) ^{2}}
\end{equation*}%
\begin{equation*}
y_{1}^{2}=\frac{\bar{\xi}_{k}\bar{\xi}_{l}\left( Nx_{1}\right) ^{2}}{\left( 
\bar{\xi}_{k}M\right) ^{2}}=\frac{1}{M^{2}}\frac{1}{\bar{\xi}_{k}}\frac{%
b\left( \sqrt{\bar{\xi}_{k}}\sqrt{\bar{\xi}_{l}}\right) ^{b}}{\left( \left( 
\sqrt{\bar{\xi}_{l}}\right) ^{b}+\left( \sqrt{\bar{\xi}_{k}}\right)
^{b}\right) ^{2}}
\end{equation*}%
\begin{eqnarray*}
U_{j} &=&\frac{\xi _{j}}{N}\left( 1-\frac{\left( y_{1}M\bar{\xi}_{k}\right)
^{b}}{\left( x_{1}N\bar{\xi}_{l}\right) ^{b}+\left( y_{1}M\bar{\xi}%
_{k}\right) ^{b}}\right) -\frac{1}{2}x_{j}^{2} \\
&=&\frac{\xi _{j}}{N}\left( 1-\frac{\left( \sqrt{\bar{\xi}_{k}}\right) ^{b}}{%
\left( \sqrt{\bar{\xi}_{l}}\right) ^{b}+\left( \sqrt{\bar{\xi}_{k}}\right)
^{b}}\right) -\frac{1}{2}\xi _{j}^{2}\frac{1}{N^{2}}\frac{1}{\bar{\xi}_{l}}%
\frac{b\left( \sqrt{\bar{\xi}_{k}}\sqrt{\bar{\xi}_{l}}\right) ^{b}}{\left(
\left( \sqrt{\bar{\xi}_{l}}\right) ^{b}+\left( \sqrt{\bar{\xi}_{k}}\right)
^{b}\right) ^{2}}
\end{eqnarray*}%
\bigskip

\begin{eqnarray*}
\frac{\partial U_{j}}{\partial \bar{\xi}_{l}} &=&\frac{1}{2N^{2}}\frac{b}{%
\bar{\xi}_{l}^{2}}\frac{\xi _{j}^{2}}{\left( \left( \sqrt{\bar{\xi}_{l}}%
\right) ^{b}+\left( \sqrt{\bar{\xi}_{k}}\right) ^{b}\right) ^{2}}\left( 
\sqrt{\bar{\xi}_{k}}\sqrt{\bar{\xi}_{l}}\right) ^{b} \\
&&+\frac{1}{2N}\frac{b}{\sqrt{\bar{\xi}_{l}}}\frac{\xi _{j}}{\left( \left( 
\sqrt{\bar{\xi}_{l}}\right) ^{b}+\left( \sqrt{\bar{\xi}_{k}}\right)
^{b}\right) ^{3}}\left( \sqrt{\bar{\xi}_{k}}\right) ^{b}\left( \sqrt{\bar{\xi%
}_{l}}\right) ^{b-1} \\
&&+\frac{1}{2N^{2}}\frac{b^{2}}{\bar{\xi}_{l}^{\frac{3}{2}}}\frac{\xi
_{j}^{2}}{\left( \left( \sqrt{\bar{\xi}_{k}}\right) ^{b}+\left( \sqrt{\bar{%
\xi}_{l}}\right) ^{b}\right) ^{3}}\left( \sqrt{\bar{\xi}_{k}}\sqrt{\bar{\xi}%
_{l}}\right) ^{b}\left( \sqrt{\bar{\xi}_{l}}\right) ^{b-1} \\
&&-\frac{1}{4N^{2}}\frac{b^{2}}{\bar{\xi}_{l}^{\frac{3}{2}}}\left( \sqrt{%
\bar{\xi}_{k}}\right) \frac{\xi _{j}^{2}}{\left( \left( \sqrt{\bar{\xi}_{l}}%
\right) ^{b}+\left( \sqrt{\bar{\xi}_{k}}\right) ^{b}\right) ^{2}}\left( 
\sqrt{\bar{\xi}_{k}}\sqrt{\bar{\xi}_{l}}\right) ^{b-1}
\end{eqnarray*}

which implies that 
\begin{eqnarray*}
\frac{\partial ^{2}U_{j}^{b}}{\partial \bar{\xi}_{B}\partial \xi _{j}} &=&%
\frac{1}{2N^{2}}\frac{b}{\bar{\xi}_{l}^{2}}\frac{2\xi _{j}}{\left( \left( 
\sqrt{\bar{\xi}_{l}}\right) ^{b}+\left( \sqrt{\bar{\xi}_{k}}\right)
^{b}\right) ^{2}}\left( \sqrt{\bar{\xi}_{k}}\sqrt{\bar{\xi}_{l}}\right) ^{b}
\\
&&+\frac{1}{2N^{2}}\frac{b^{2}}{\bar{\xi}_{l}^{\frac{3}{2}}}\frac{2\xi _{j}}{%
\left( \left( \sqrt{\bar{\xi}_{k}}\right) ^{b}+\left( \sqrt{\bar{\xi}_{l}}%
\right) ^{b}\right) ^{3}}\left( \sqrt{\bar{\xi}_{k}}\sqrt{\bar{\xi}_{l}}%
\right) ^{b}\left( \sqrt{\bar{\xi}_{l}}\right) ^{b-1} \\
&&+\frac{1}{4N^{2}}\frac{b^{2}}{\bar{\xi}_{l}^{\frac{3}{2}}}2\xi _{j}\frac{%
\left( \sqrt{\bar{\xi}_{l}}\right) ^{b}-\left( \sqrt{\bar{\xi}_{k}}\right)
^{b}}{\left( \left( \sqrt{\bar{\xi}_{l}}\right) ^{b}+\left( \sqrt{\bar{\xi}%
_{k}}\right) ^{b}\right) ^{2}}\left( \sqrt{\bar{\xi}_{k}}\sqrt{\bar{\xi}_{l}}%
\right) ^{b}
\end{eqnarray*}

and thus the condition \ref{endogenous segregation condition} is satisfied
so that we can have an analog of Proposition \ref{segregate decentralized}. $%
\blacksquare $


\begin{thebibliography}{99}
\bibitem{Ana} Ananat, Elizabeth Oltmans. ``The wrong side (s) of the tracks:
The causal effects of racial segregation on urban poverty and inequality."
American Economic Journal: Applied Economics 3.2 (2011): 34-66.

\bibitem{AIMR} Aimone, Jason A., et al. ``Endogenous group formation via
unproductive costs." Review of Economic Studies 80.4 (2013): 1215-1236.

\bibitem{BY} Baccara, Mariagiovanna, and Leeat Yariv. ``Homophily in peer
groups." American Economic Journal: Microeconomics 5.3 (2013): 69-96.

\bibitem{BO} Baumol, William J., et William Oates. The theory of
environmental policy. Cambridge university press, 1988.

\bibitem{BPD} Ben-Porath, Elchanan, and Eddie Dekel. ``Signaling future
actions and the potential for sacrifice." Journal of Economic Theory 57.1
(1992): 36-51.

\bibitem{Berg} Berglas, Eitan. ``On the theory of clubs." The American
Economic Review 66.2 (1976): 116-121.

\bibitem{Ber} Berman, Eli. 2000. ``Sect, Subsidy, and Sacrifice: An
Economist's View of Ultra-Orthodox Jews." Quarterly Journal of Economics 115
(3): 905-53.

\bibitem{Buc} Buchanan, James M. ``An economic theory of clubs." Economica
32.125 (1965): 1-14.

\bibitem{Cim} Cimino, Aldo. ``The evolution of hazing: Motivational
mechanisms and the abuse of newcomers." Journal of Cognition and Culture
11.3-4 (2011): 241-267.

\bibitem{Cle} Cleaver, Frances. ``The inequality of social capital and the
reproduction of chronic poverty." World development 33.6 (2005): 893-906.

\bibitem{CS} Cornes, Richard C., and Emilson Delfino Silva. ``Prestige
Clubs." University of Alberta School of Business Research Paper 2013-1315
(2013).

\bibitem{CJL} Cullen, Julie Berry, Brian A. Jacob, and Steven Levitt. ``The
effect of school choice on participants: Evidence from randomized
lotteries." Econometrica 74.5 (2006): 1191-1230.

\bibitem{DDK} Duflo, Esther, Pascaline Dupas, and Michael Kremer. ``Peer
effects, teacher incentives, and the impact of tracking: Evidence from a
randomized evaluation in Kenya." American Economic Review 101.5 (2011):
1739-74.

\bibitem{DS } Durlauf, Steven N., and Ananth Seshadri. ``Is assortative
matching efficient?." Economic Theory 21.2-3 (2003): 475-493.

\bibitem{ER} Epple, Dennis, and Richard E. Romano. ``Competition between
private and public schools, vouchers, and peer-group effects." American
Economic Review (1998): 33-62.

\bibitem{GYR} Graham, Bryan S., Guido W. Imbens, and Geert Ridder.
``Complementarity and aggregate implications of assortative matching: A
nonparametric analysis." Quantitative Economics 5.1 (2014): 29-66.

\bibitem{GGKS} Greenwood, J., Guner, N., Kocharkov, G., \& Santos, C.
(2014). ``Marry your like: Assortative mating and income inequality.''\
American Economic Review, 104(5), 348-53.

\bibitem{GGL} Groves, Mark, Gerald Griggs, and Kathryn Leflay. ``Hazing and
initiation ceremonies in university sport: setting the scene for further
research in the United Kingdom." Sport in Society 15.1 (2012): 117-131.

\bibitem{HMS} Hoppe, Heidrun C., Benny Moldovanu, and Aner Sela. ``The
theory of assortative matching based on costly signals." The Review of
Economic Studies 76.1 (2009): 253-281.

\bibitem{HS} Helsley, Robert W., and William C. Strange. ``Exclusion and the
Theory of Clubs." Canadian Journal of Economics (1991): 888-899.

\bibitem{HU} Hsieh, Chang-Tai, and Miguel Urquiola. ``The effects of
generalized school choice on achievement and stratification: Evidence from
Chile's voucher program." Journal of public Economics 90.8-9 (2006):
1477-1503.

\bibitem{Ian} Iannaccone, Laurence R. ``Sacrifice and stigma: Reducing
free-riding in cults, communes, and other collectives." Journal of political
economy 100.2 (1992): 271-291.

\bibitem{JM} Jaramillo, Fernando, and Fabien Moizeau. ``Conspicuous
consumption and social segmentation." Journal of Public Economic Theory 5.1
(2003): 1-24.

\bibitem{JMS} Jenkins, Stephen P., John Micklewright, and Sylke V. Schnepf.
``Social segregation in secondary schools: how does England compare with
other countries?." Oxford Review of Education 34.1 (2008): 21-37.

\bibitem{JBK} Jiang, Yuexin, Daniel I. Bolnick, and Mark Kirkpatrick.
``Assortative mating in animals." The American Naturalist 181.6 (2013):
E125-E138.

\bibitem{Keat} Keating, Caroline F., et al. ``Going to College and Unpacking
Hazing: A Functional Approach to Decrypting Initiation Practices Among
Undergraduates." Group Dynamics: Theory, Research, and Practice 9.2 (2005):
104.

\bibitem{LN} Legros, Patrick, and Andrew F. Newman. ``Beauty is a beast,
frog is a prince: Assortative matching with nontransferabilities."
Econometrica 75.4 (2007): 1073-1102.

\bibitem{LR} Levy, Gilat, and Ronny Razin. ``Religious beliefs, religious
participation, and cooperation." American economic journal: microeconomics
4.3 (2012): 121-51.

\bibitem{DK} de Klerk, Vivian. ``Initiation, hazing or orientation? A case
study at a South African university." International Research in Education
1.1 (2013): 86-100.

\bibitem{MU} MacLeod, W. Bentley, and Miguel Urquiola. ``Reputation and
school competition." American Economic Review 105.11 (2015): 3471-88.

\bibitem{MMF} Mercuro, Anne, Samantha Merritt, and Amanda Fiumefreddo. ``The
Effects of Hazing on Student Self-Esteem: Study of Hazing Practices in Greek
Organizations in a State College." The Ramapo Journal of Law and Society
(2014).

\bibitem{Oak} Oakland, William H. ``Congestion, public goods and welfare."
Journal of Public Economics 1.3-4 (1972): 339-357.

\bibitem{OR} {\O }stvik, Kristina, and Floyd Rudmin. ``Bullying and hazing
among Norwegian army soldiers: Two studies of prevalence, context, and
cognition." Military psychology 13.1 (2001): 17-39.

\bibitem{PPU} Page, Talbot, Louis Putterman, and Bulent Unel. ``Voluntary
association in public goods experiments: Reciprocity, mimicry and
efficiency." The Economic Journal 115.506 (2005): 1032-1053.

\bibitem{Rear} Reardon, Sean F. ``School segregation and racial academic
achievement gaps." RSF: The Russell Sage Foundation Journal of the Social
Sciences 2.5 (2016): 34-57.

\bibitem{RS} Ruffle, Bradley J., and Richard Sosis. ``Does it pay to pray?
Costly ritual and cooperation." The BE Journal of Economic Analysis \&
Policy 7.1 (2007).

\bibitem{SK} Skaperdas, Stergios. ``Contest success functions." Economic
theory 7.2 (1996): 283-290.

\bibitem{Sol} Soler, Montserrat. ``Costly signaling, ritual and cooperation:
evidence from Candombl{\'e}, an Afro-Brazilian religion." Evolution and
Human Behavior 33.4 (2012): 346-356.

\bibitem{Sos} Sosis, Richard. ``The adaptive value of religious ritual:
Rituals promote group cohesion by requiring members to engage in behavior
that is too costly to fake." American scientist 92.2 (2004): 166-172.

\bibitem{SKB} Sosis, Richard, Howard C. Kress, and James S. Boster. ``Scars
for war: Evaluating alternative signaling explanations for cross-cultural
variance in ritual costs." Evolution and Human Behavior 28.4 (2007): 234-247.

\bibitem{Vig} Vigil, James. ``Street baptism: Chicano gang initiation."
Human Organization 55.2 (1996): 149-153.

\bibitem{Win} Winslow, Donna. ``Rites of passage and group bonding in he
Canadian Airborne." Armed Forces \& Society 25.3 (1999): 429-457.
\end{thebibliography}
\end{document}